\documentclass[10pt,conference]{IEEEtran}

\IEEEoverridecommandlockouts
\usepackage{amsfonts}
\usepackage[dvips]{graphicx}
\usepackage{times}
\usepackage{cite}
\usepackage{amsmath}
\usepackage{cases}
\usepackage{array}
\usepackage{dsfont}
\usepackage{amssymb}

\usepackage{stfloats}
\usepackage{slashbox}
\usepackage{graphicx}
\usepackage{footnote}
\usepackage{booktabs}
\usepackage{array}
\usepackage{multirow}
\usepackage{bm}
\usepackage{empheq}
\usepackage[labelformat=simple]{subcaption}

\usepackage{amsthm}
\usepackage{algorithm}
\usepackage{algorithmic}
\usepackage{color}
\theoremstyle{plain}
\newtheorem{theorem}{Theorem}
\newtheorem{corollary}{Corollary}

\theoremstyle{definition}

\newtheorem{definition}{Definition}
\newtheorem{remark}{Remark}
\begin{document}

\abovedisplayskip=3pt
\belowdisplayskip=3pt
%\title{}
%\author{\IEEEauthorblockN{Fan Xu, Meixia Tao, \IEEEmembership{Senior Member,~IEEE}}}
%\maketitle

\title{Fundamental Limits of Decentralized Caching in Fog-RANs with Wireless Fronthaul}
\author{\IEEEauthorblockN{Fan Xu and Meixia Tao}\\
\IEEEauthorblockA{Department of Electronic Engineering, Shanghai Jiao Tong University, Shanghai, China \\Emails: \{xxiaof, mxtao\}@sjtu.edu.cn}\\
\thanks{This work is supported by the National Natural Science Foundation of China under grants 61571299 and 61521062.}}
\maketitle

\begin{abstract}
This paper aims to characterize the synergy of distributed caching and wireless fronthaul in a fog radio access network (Fog-RAN) where all edge nodes (ENs) and user equipments (UEs) have a local cache and store contents independently at random. The network operates in two phases, a file-splitting based decentralized cache placement phase and a fronthaul-aided content delivery phase. We adopt normalized delivery time (NDT) to characterize the asymptotic latency performance with respect to cache size and fronthaul capacity. Both an achievable upper bound and a theoretical lower bound of NDT are obtained, and their multiplicative gap is within 12. In the proposed delivery scheme, we utilize the fronthaul link, by exploiting coded multicasting, to fetch both non-cached and cached contents to boost EN cooperation in the access link. In particular, to fetch contents already cached at ENs, an additional layer of coded multicasting is added on the coded messages desired by UEs in the fronthaul link. Our analysis shows that the proposed delivery scheme can balance the delivery latency between the fronthaul link and access link, and is approximately optimum under decentralized caching.
\end{abstract}

\section{Introduction}
Caching is emerging as an effective technique to reduce peak-hour data traffic and improve user perceived experience in wireless networks. Unlike traditional web-caching and in-network caching, caching at the edge of wireless networks is able to exploit the broadcast nature of wireless medium and thus achieve global caching gain \cite{fundamentallimits}. Recently, it has attracted many interests to characterize the fundamental limits of caching in various wireless networks. This work aims to advance this topic by studying the synergy between distributed caching and wireless fronthauling in fog radio access networks (Fog-RANs).

Previously, the gain of caching is studied in wireless interference networks where caches are equipped at all transmitters and receivers  \cite{mine,bothcache,niesen,gunduz}. It is found in \cite{mine} that with a generic file splitting and caching strategy, the interference network topology can be changed into a new family of channels, referred to as \emph{cooperative X-multicast channels}, and hence leverage transmitter cooperation gain and coded multicasting gain, apart from receiver local caching gain. These works \cite{mine,bothcache,niesen,gunduz}, however, have assumed that there exists a central controller that coordinates the file splitting and cache placement among all distributed nodes (at least all the transmit nodes if not receive nodes as in \cite{gunduz}). Moreover, they also assume that the total cache size among the network nodes is large enough to collectively store the entire database without cache miss.

To exploit the potential of practical and scalable caching in large and random networks, decentralized coded caching is proposed at the user side where users can independently cache file bits at random \cite{decentralized}. It is shown in \cite{decentralized} that decentralized coded caching can achieve a performance close to the optimal centralized scheme.

The idea of decentralized coded caching can be extended to a general interference network where all the transmitters and receivers cache file bits independently at random. However, due to the lack of a central controller for careful cache placement, it is very likely to have cache miss even when the total cache size is as large as in the centralized scheme.  To overcome the cache miss issue, the works \cite{simeone},\cite{Koh2017Cloud} consider a Fog-RAN where each cache-enabled edge node (EN) is connected  via a fronthaul link to a cloud server which has access to the entire database. They characterize a latency-oriented performance with respect to both the EN cache size and the fronthaul capacity. The works \cite{Cran,Joan2018storage,girgis2017decentralized} consider a Fog-RAN where all ENs and user equipments (UEs) are equipped with local caches. The authors in \cite{Cran} propose a caching-and-delivery scheme that combines network-coded fronthaul transmission and cache-aided interference management. The authors in \cite{Joan2018storage} propose a mixed cache placement, i.e., centralized caching at ENs and decentralized caching at UEs, and employ a combination of interference management techniques in the delivery phase. The authors in \cite{girgis2017decentralized} consider a decentralized cache placement at all ENs and UEs and propose a coded delivery strategy that exploits the network topology for Fog-RANs. Note that \cite{girgis2017decentralized} is only limited to two ENs only.

The contribution of this work is to characterize the latency performance of a Fog-RAN with wireless fronthaul and for arbitrary number of ENs and UEs, where all ENs and UEs are equipped with caches. Considering the random mobility of UEs and the dynamic on/off of ENs, we apply decentralized cache placement at all ENs and UEs  without central coordination. As in \cite{simeone,mine,gunduz,Cran,girgis2017decentralized,Joan2018storage,Koh2017Cloud}, we adopt normalized delivery time (NDT) as the performance metric. The network operates in two phases, a decentralized cache placement phase and a fronthaul-aided content delivery phase. In our proposed delivery scheme, the wireless fronthaul is not only responsible to fetch cache-miss contents but also can be used to fetch contents already cached at ENs to boost transmission cooperation to any desired level in the access link. To fetch contents already cached at ENs, an additional layer of coded multicasting on top of the coded messages desired by UEs is exploited in the fronthaul link. To fetch contents not cached at ENs, the coded messages, rather than the original files desired by UEs, are transmitted in the fronthaul link. The access transmission in our proposed delivery scheme is similar to \cite{mine}, which transforms the access link into the cooperative X-multicast channel. Based on the proposed delivery scheme, we obtain an achievable upper bound of the minimum NDT of the network with decentralized caching. Numerical results show that our NDT performance is even better than that using centralized caching \cite{Cran} with wireless fronthaul and that using centralized caching \cite{simeone}, mixed caching \cite{Joan2018storage}, and decentralized caching \cite{girgis2017decentralized} with dedicated fronthaul under certain conditions. Under decentralized caching, we also obtain a theoretical lower bound of the minimum NDT by applying cut-set-like bounds in the fronthaul transmission and access transmission separately. It is shown that the multiplicative gap between the upper and lower bounds is within 12.

Notations: $[K]$ denotes the set $\{1,2,\ldots,K\}$. $\mathcal{CN}(0,1)$ denotes the complex-valued Gaussian distribution with zero mean and unit variance.

\section{System Model and Problem Description}
\subsection{Fog-RAN with Wireless Fronthaul}
We consider a Fog-RAN as shown in Fig~\ref{Fig model}, where there are $N_T$ ($N_T\!\ge\!2$) ENs, $N_R$ ($N_R\!\ge\!2$) UEs, and the ENs are connected to a macro base station (MBS), or a cloud server, through a shared wireless fronthaul link.  All ENs and UEs have a local cache each.  The access link between each EN and each UE experiences channel fading, and is corrupted with additive white Gaussian noise. The communication at each time slot $t$ over the access channel is modeled by

\begin{align}
Y_q(t)=\sum_{p=1}^{N_T}h_{qp}(t)X_p(t)+Z_q(t),q\in[N_R],\notag
\end{align}
where $Y_q(t)\!\in\! \mathbb{C}$ is the received signal at UE $q$, $X_p(t)\!\in\! \mathbb{C}$ is the transmitted signal at EN $p$, $h_{qp}(t)\!\in \!\mathbb{C}$ is the channel coefficient from EN $p$ to UE $q$ which is assumed to be independent and identically distributed (i.i.d.) as some continuous distribution, and $Z_q(t)$ is the noise at UE $q$ distributed as $\mathcal{CN}(0,1)$.

The fronthaul link between the MBS and all the ENs also experiences channel fading and additive white Gaussian noise. The communication at each time slot $t$ over the fronthaul channel is modeled by

\begin{align}
Q_p(t)=g_{p}(t)S(t)+N_p(t),p\in[N_T],\notag
\end{align}
where $Q_p(t)\!\in\! \mathbb{C}$ is the received signal at EN $p$, $S(t)\!\in\! \mathbb{C}$ is the transmitted signal from the MBS, $g_{p}(t)\!\in\! \mathbb{C}$ is the channel coefficient from the MBS to EN $p$ which is assumed to be i.i.d. as some continuous distribution, and $N_p(t)$ is the noise at EN $p$ distributed as $\mathcal{CN}(0,1)$.

Consider a database consisting of $N$ files, denoted as $\{W_1,W_2,\ldots,W_N\}$, each with $F$ bits. Throughout this study, we consider $N\ge N_R$ so that each UE can request a distinct file. The MBS has full access to the database via a dedicated backhaul link. Each EN can store $\mu_TNF$  ($\mu_T\le1$) bits locally, and each UE can store $\mu_RNF$ ($\mu_R\le1$) bits locally, where $\mu_T$ and $\mu_R$ are referred to as \textit{normalized cache sizes} at each EN and UE, respectively. In this work, we consider the complete region for the normalized cache sizes, i.e., $0 \le \mu_R, \mu_T \le 1$, because of the presence of the fronthaul.  Note that the works \cite{mine,gunduz} only consider the feasible region  $N_T\mu_T+\mu_R\ge1$, while \cite{bothcache,niesen} only consider $N_T\mu_T\ge1$.

\begin{figure}[!tbp]
\begin{centering}
\includegraphics[scale=0.19]{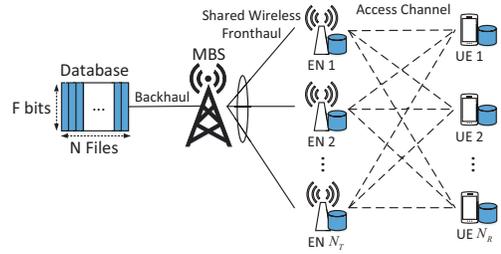}
\caption{Cache-aided Fog-RAN.}\label{Fig model}
\end{centering}
\end{figure}

The network operates in two phases, a \textit{decentralized cache placement phase} and \textit{a two-hop content delivery phase}, as detailed in the next two subsections.

\subsection{Decentralized Cache Placement}
We adopt decentralized cache placement at all ENs and UEs. Each EN $p$ (or UE $q$) independently caches a subset of $\mu_TF$ (or $\mu_RF$) bits of each file $W_n$, chosen uniformly at random, denoted as $U_p^n$ (or $V_q^n$), for $n\in [N]$. Denote $U_p\!\triangleq\!\bigcup_{n\in [N]}\! U_{p}^n$ (or $V_q\!\triangleq\!\bigcup_{n\in [N]}\! V_{q}^n$) as all the cached  bits at EN $p$ (or UE $q$). Note that neither inter-file nor intra-file coding is allowed in the considered decentralized cache placement phase.

%Consider an arbitrary UE set $\Phi$ with $m$ UEs ( $0\le m\le N_R$) and an arbitrary EN set $\Psi$ with $n$ ENs ($0\le n \le N_T$). Given that each UE and EN independently caches a subset of $\mu_RF$ and $\mu_TF$ bits of each file uniformly at random, respectively, the probability of each bit of the files that is cached exactly the the UE set $\Phi$ and  EN set $\Psi$ is $\mu_R^m(1-\mu_R)^{N_R-m}\mu_T^n(1-\mu_T)^{N_T-n}$. Therefore, based on the law of large numbers, the size of  the subfile of each file that is cached exactly at the UE set $\Phi$ and  EN set $\Psi$ is $\mu_R^m(1-\mu_R)^{N_R-m}\mu_T^n(1-\mu_T)^{N_T-n}F+o(F)$ bits with probability approaching one, similar to \cite[Appendix A]{decentralized}. For ease of explanation, we ignore the $o(F)$ term in the following analysis. We define
%\begin{align}
%f_{m,n}\!\triangleq\!\mu_R^m(1-\mu_R)^{N_R-m}\mu_T^n(1-\mu_T)^{N_T-n}\label{eqn subfile size}
%\end{align}
%as the fractional size of each subfile cached exactly at UE set $\Phi$ with $m$ UEs and EN set $\Psi$ with $n$ ENs.

By the law of large numbers, when file size $F$ is large enough, the size of each subfile cached exactly at an arbitrary set of $m$ UEs ($0\!\le\! m\!\le\! N_R$) and an arbitrary set of $n$ ENs ($0\!\le\! n\! \le\! N_T$) is $\mu_R^m(1-\mu_R)^{N_R-m}\mu_T^n(1-\mu_T)^{N_T-n}F+o(F)$ bits with high probability. Since this paper focuses on the extreme case when $F\rightarrow \infty$, we ignore the $o(F)$ term in the rest of the paper, and define
\begin{align}
f_{m,n}\!\triangleq\!\mu_R^m(1-\mu_R)^{N_R-m}\mu_T^n(1-\mu_T)^{N_T-n}\label{eqn subfile size}
\end{align}
as the fractional size of each subfile cached exactly at each node set with $m$ UEs and $n$ ENs, similar to \cite{decentralized}.

\subsection{Two-Hop Content Delivery}
Each UE $q$ requests a file $W_{d_q}$ from the database. We denote ${\bf d}\triangleq(d_q)^{N_R}_{q=1}\in[N]^{N_R}$ as the demand vector from all the $N_R$ UEs.
The content delivery phase is a two-hop transmission process, the first hop being the fronthaul channel and the second hop being the access channel.

\subsubsection{Fronthaul Transmission}\label{section fronthaul model}
The MBS employs an encoding function $\Lambda_C$ to map the entire database, UE demand ${\bf d}$, and channel realization $\mathbf{G}_F\!\triangleq\! \{g_{p}(t)\!:\!\forall p \!\in\! [N_T],\forall t\!\in\! [T_F]\}$ to a length-$T_F$ codeword $(S[t])^{T_F}_{t=1}$ with an average transmit power constraint $P^r$, i.e., $ \frac{1}{T_F}\!\sum_{t=1}^{T_F}\!|S(t)|^2\!\le\! P^r$, where $r\!>\!0$ is the power scaling of the fronthaul link compared to the access link with power $P$. Note that $r$ can also be viewed as the multiplexing gain in the fronthaul link \cite{girgis2017decentralized}.

\subsubsection{Access Transmission}
In this paper, we assume that all the ENs are half-duplex, which means that they cannot transmit over the access link while receiving from the fronthaul link at the same time. Thus, the fronthaul transmission and the access transmission take place in serial. After receiving signals from the fronthaul link, each EN $p$ uses an encoding function $\Lambda_p$ to map its cached content $U_p$, UE demand ${\bf d}$, received signals $(Q_p[t])^{T_F}_{t=1}$, and channel realizations $\mathbf{G}_F$ and $\mathbf{H}_A\triangleq \{h_{qp}(t):\forall q \in [N_R], \forall p \in [N_T],\forall t\in [T_A]\}$ to a length-$T_A$ codeword $(X_p[t])^{T_A}_{t=1}$. Note that $T_F$ and $T_A$ may depend on the UE demand ${\bf d}$ and channel realizations $\mathbf{G}_F$ and $\mathbf{H}_A$. Each codeword $(X_p[t])_{t=1}^{T_A}$ has an average transmit power constraint $P$, i.e., $\frac{1}{T_A}\sum_{t=1}^{T_A}|X_p(t)|^2\le P$.

Upon receiving signals $(Y_q[t])^{T_A}_{t=1}$ in the access link, each UE $q$ employs a decoding function $\Gamma_q$ to decode $\hat{W}_{d_q}$ of its desired file $W_{d_q}$ from $(Y_q[t])^{T_A}_{t=1}$ along with its cached content $V_q$, UE demand ${\bf d}$, and channel realizations $\mathbf{G}_F$ and $\mathbf{H}_A$ as side information.

Define $P_\epsilon\triangleq\max_\mathbf{d}\max_{q}\mathbb{P}(\hat{W}_{d_q}\ne W_{d_q})$ as the worst-case error probability. A given set of coding functions $\{\Lambda_C, \Lambda_p,\Gamma_q: p\in[N_T],q\in[N_R]\}$ in the delivery phase is said to be feasible if, for almost all channel realizations, $P_\epsilon\to 0$ when $F\to\infty$.

\subsection{Performance Metric}\label{section ndt}
Following \cite{simeone}, we adopt \textit{normalized delivery time} (NDT) as the performance metric which is given by\footnote{The same metric is also defined in \cite{zhangElia} but under a different name.}
\begin{align}
\tau(\mu_R,\mu_T,r)\triangleq\lim_{P\to\infty}\lim_{F\to\infty}\sup\frac{\max_\mathbf{d}(T_F+T_A)}{F/\log P}.\label{eqn ndt definition}
\end{align}
We are interested in characterizing the minimum NDT of the network with decentralized caching which is defined as
\begin{align}
\tau^*(\mu_R,\mu_T,r)\!=\!\inf\{\tau(\mu_R,\mu_T,r)\!:\!\tau(\mu_R,\mu_T,r)\!\textrm{ is\! achievable}\}.\notag
\end{align}

\begin{remark}\label{remark ndt explaination}
Similar to \cite{Cran}, given the two-hop content delivery phase, the NDT can be rewritten as $\tau=\tau_F+\tau_A$, where $\tau_F$ and $\tau_A$ is the NDT in the fronthaul link and access link, respectively. Based on the power constraint of fronthaul codeword $(S[t])^{T_F}_{t=1}$, the fronthaul link carries $r\log P$ bits per channel use in the high SNR regime. Denote $R_F$ as the sum traffic load normalized by file size $F$ in the fronthaul link. We can rewrite $\tau_F$ as $\tau_F=R_F/r$. Similar to \cite[Remark 1]{mine}, denote $R_A$ as the per-user traffic load normalized by file size $F$ and $d$ as the per-user degrees of freedom (DoF) in the access link. We can rewrite $\tau_A$ as $\tau_A=R_A/d$. Therefore, the NDT can be expressed more conveniently as
\begin{align}
\tau=\tau_F+\tau_A=R_F/r+R_A/d.  \label{eqn ndt explaination}
\end{align}
\end{remark}

\section{Delivery Scheme for $3\times3$ Fog-RAN}\label{section delivery33}
In this section, we use a Fog-RAN with $N_T=N_R=3$ to illustrate the proposed delivery scheme. The scheme can be easily generated to a general Fog-RAN with arbitrary $N_T$ and arbitrary $N_R$, which is given in Section \ref{section upper bound}. We consider the worst-case scenario that each UE requests a distinct file. Note that when some UEs request the same file, the proposed delivery scheme can still be applied by treating the requests as being different. Without loss of generality, we assume that UE $q$  desires $W_q$, for $q\in[3]$. We denote $W_{q,\Phi,\Psi}$ as the subfile desired by UE $q$ and cached at UE set $\Phi$ and EN set $\Psi$. Its fractional size  is given by $f_{|\Phi|,|\Psi|}=\mu_R^{|\Phi|}(1-\mu_R)^{3-|\Phi|}\mu_T^{|\Psi|}(1-\mu_T)^{3-|\Psi|}$ based on \eqref{eqn subfile size}.

Excluding the locally cached subfiles,  each UE $q$, for $q\in[3]$, wants subfiles $\{W_{q,\Phi,\Psi}\!:\!\Phi\!\not\ni \!q,\Phi\!\subseteq\![3],\Psi\!\subseteq\![3]\}$. We divide the subfiles wanted by all UEs into different groups according to the size of $\Phi$ and $\Psi$, indexed by $\{(m,n)\!:\!m\!\in\![2]\cup\{0\}, n\!\in\![3]\cup\{0\}\}$, such that subfiles in group $(m,n)$ are cached at $m$ UEs and $n$ ENs. There are $3\binom{2}{m}\binom{3}{n}$ subfiles in group $(m,n)$. Each group of subfiles is delivered individually in the time division manner. In the following, we present the delivery strategy of  two representative groups, $(m,0)$ and $(m,1)$, where $m\in[2]\cup\{0\}$. Before that, let us introduce the cooperative X-multicast channel defined in \cite[Definition 2]{mine} which shall be mentioned throughout this section and Section \ref{section upper bound}.

\begin{definition}[\cite{mine}]\label{def channel}
The channel characterized as follows is referred to as the $\binom{N_T}{j}\times\binom{N_R}{m+1}$ cooperative X-multicast channel:
\begin{enumerate}
  \item there are $N_R$ UEs and $N_T$ ENs;
  \item each set of $m+1$ ($m< N_R$) UEs forms a UE multicast group;
  \item each set of $j$ ($j\le N_T$) ENs forms a EN cooperation group;
  \item each EN cooperation group has an independent message for each UE multicast group.
\end{enumerate}
\end{definition}

\subsection{Delivery of Group $(m,0)$}\label{section delivery33 m0}
Each subfile in group $(m,0)$ is desired by one UE, cached at $m$ other UEs  but none of ENs. Coded multicasting can be exploited through bit-wise XOR, similar to \cite{fundamentallimits}. In specific, the set of coded messages is given by
\begin{align}
  \bigg\{\!W_{\Phi^+,\emptyset}^\oplus\!\triangleq\! \bigoplus_{q\in\Phi^+}\! W_{q,\Phi^+\backslash \{q\},\emptyset}\!:\!\Phi^+\!\subseteq\![3],|\Phi^+|\!=\!m+1\!\bigg\}.\label{eqn example m0 1}
\end{align}
In this work, we focus on the case with $F \rightarrow \infty$ for analytical tractability.  By the law of large numbers, each coded message  $W_{\Phi^+,\emptyset}^\oplus$  has $f_{m,0}F$ bits, and is desired by UE set $\Phi^+$. Since the ENs do not have the coded messages in \eqref{eqn example m0 1}, these messages need to be generated at the MBS and then delivered to UEs via two compulsory hops, the fronthaul link and the access link. In the fronthaul link, we let the MBS naively multicast each coded message in \eqref{eqn example m0 1}  one by one to all three ENs. Thus, from Remark \ref{remark ndt explaination}, the NDT of the fronthaul link is given by
\begin{align}
  \tau_F=\frac{\binom{3}{m+1}f_{m,0}}{r}.\label{eqn 33m0 fronthaul}
\end{align}

By such naive multicasting in the fronthaul link, each EN now has access to all the coded messages in \eqref{eqn example m0 1}, and can transmit with full cooperation in the access link. The access channel thus becomes the $\binom{3}{3}\times\binom{3}{m+1}$ cooperative X-multicast channel in Definition \ref{def channel}, whose achievable per-user DoF is $d_{m,3}=1$ in \cite[Lemma 1]{mine}. Since each UE desires $\binom{2}{m}$ coded messages, from Remark \ref{remark ndt explaination}, the NDT of the access link is given by
\begin{align}
  \tau_A=\frac{\binom{2}{m}f_{m,0}}{d_{m,3}}=\binom{2}{m}f_{m,0}.\label{eqn 33m0 access}
\end{align}

Summing up \eqref{eqn 33m0 fronthaul} and \eqref{eqn 33m0 access}, the total NDT for group $(m,0)$ is
\begin{equation}
  \tau_{m,0}=\frac{\binom{3}{m+1}f_{m,0}}{r}+\binom{2}{m}f_{m,0}.\notag
\end{equation}

\subsection{Delivery of Group $(m,1)$}\label{section delivery33 m1}
Unlike the subfiles in group $(m,0)$, each subfile in group $(m,1)$ is already cached at one EN, and therefore the coded messages can be generated at each EN locally. In specific, each EN $p$, for $p\in[3]$, generates:
\begin{align}
   \bigg\{\!W_{\Phi^+\!,\{p\}}^\oplus\!\triangleq\! \bigoplus_{q\in\Phi^+}\! W_{q,\Phi^+\backslash \{q\},\{p\}}\!\!:\!\Phi^+\!\subseteq\![3],\!|\Phi^+\!|\!=\!m\!+\!1\!\bigg\}\!.\!\label{eqn example m1 1}
\end{align}
Each coded message $W_{\Phi^+,\{p\}}^\oplus$ has $f_{m,1}F$ bits, and is desired by UE set $\Phi^+$. These coded messages can be delivered to UEs via one hop in the access link without the use of fronthaul link or delivered via two hops with the aid of fronthaul link.

\subsubsection{Without Fronthaul}
Each EN $p$, for $p\in[3]$, sends $\{W_{\Phi^+,\{p\}}^\oplus\}$ in the access link, and the access channel becomes the $\binom{3}{1}\times\binom{3}{m+1}$ cooperative X-multicast channel with achievable per-user DoF $d_{m,1}$  in \cite[Lemma 1]{mine}. Since each UE desires $3\binom{2}{m}$ messages, the NDT is given by
\begin{align}
  \tau=\frac{3\binom{2}{m}f_{m,1}}{d_{m,1}}.\label{eqn 33m1 nofronthaul}
\end{align}
\subsubsection{With Fronthaul}
With the aid of fronthaul, we can allow ENs to access the coded messages of others via the transmission of the MBS in the fronthaul link, thereby enabling transmission cooperation among ENs in the access link. As a price to pay for the EN cooperation gain, additional fronthaul delivery latency is caused.  Thus, the optimal cooperation strategy should balance the time between the access link and the fronthaul link.

Assume that after the aid of fronthaul transmission, every set of $1+i$ ENs can form a cooperation group in the access link, where $i\!\in\! [2]$ is a design parameter to balance the tradeoff mentioned above. We split each message $W_{\Phi^+,\{p\}}^\oplus$ in \eqref{eqn example m1 1} into $\binom{2}{i}$ sub-messages $\{W_{\Phi^+,\{p\}}^{\oplus,\Psi^+}:\Psi^+\subseteq[3],|\Psi^+|=1+i,\Psi^+\ni p\}$, each with $f_{m,1}/\binom{2}{i}F$ bits and sent by EN set $\Psi^+$ exclusively in the access transmission. Consider an arbitrary EN set $\Psi^+$ with size $1+i$. ENs in $\Psi^+$ need to send sub-messages
\begin{align}
\left\{W_{\Phi^+,\{p\}}^{\oplus,\Psi^+}:\Phi^+\subseteq[3],|\Phi^+|=m+1,p\in\Psi^+\right\}\label{eqn 33m1 fronthaul 1}
\end{align}
to UEs, and the MBS needs to send $W_{\Phi^+,\{p\}}^{\oplus,\Psi^+}$ to ENs $\{p':p'\in\Psi^+\backslash \{p\}\}$ which do not cache it. Given that each sub-message is already cached at one EN, coded multicasting can be used in the fronthaul transmission. In specific, the MBS sends coded sub-messages
\begin{align}
  \left\{ W_{\Phi^+,\{p\}}^{\oplus,\Psi^+}\oplus W_{\Phi^+,\{p'\}}^{\oplus,\Psi^+}\!:\!\Phi^+\subseteq[3],|\Phi^+|=m+1,p,p'\in\Psi^+\right\}\notag
\end{align}
to EN set $\Psi^+$. Note that in \cite{Cran}, the MBS sends coded messages generated directly from subfiles $\{W_{q,\Phi,\{p\}}\}$ to ENs, not from the coded messages in \eqref{eqn 33m1 fronthaul 1}. Compared to \cite{Cran}, an additional layer of XOR combining on top of the coded messages desired by UEs is exploited in our scheme. Upon receiving the above  coded sub-messages, each EN in $\Psi^+$ can decode its desired sub-messages with  its local cache. The fronthaul NDT for the given $i$ is  thus given by
\begin{equation}
  \tau_F=\frac{\binom{3}{m+1}\binom{3}{1+i}\binom{1+i}{2}f_{m,1}}{r\binom{2}{i}}=\frac{3\binom{3}{m+1}if_{m,1}}{2r}.\label{eqn 33m1 fronthaul 2}
\end{equation}

By such coded multicasting in the fronthaul link, each EN set $\Psi^+$ with size $1+i$ can cooperatively send sub-messages in \eqref{eqn 33m1 fronthaul 1}, each desired by $m+1$ UEs. The network in the access transmission is upgraded into the $\binom{3}{1+i}\times\binom{3}{m+1}$ cooperative X-multicast channel with achievable per-user DoF $d_{m,1+i}$ in \cite[Lemma 1]{mine}. Since each UE wants $\binom{2}{m}\binom{3}{1+i}\binom{1+i}{1}$ sub-messages, each with $f_{m,1}/\binom{2}{i}F$ bits, the access NDT is
\begin{equation}
  \tau_A=\frac{\binom{2}{m}\binom{3}{1+i}\binom{1+i}{1}f_{m,1}}{\binom{2}{i}d_{m,1+i}}=\frac{3\binom{2}{m}f_{m,1}}{d_{m,1+i}}.\label{eqn 33m1 fronthaul 3}
\end{equation}
Summing up \eqref{eqn 33m1 fronthaul 2} and \eqref{eqn 33m1 fronthaul 3}, the total NDT is given by
\begin{equation}
  \tau=\frac{3\binom{3}{m+1}if_{m,1}}{2r}+\frac{3\binom{2}{m}f_{m,1}}{d_{m,1+i}}.\label{eqn 33m1 fronthaul 4}
\end{equation}

Choosing the smallest NDT among \eqref{eqn 33m1 nofronthaul} without fronthaul delivery and \eqref{eqn 33m1 fronthaul 4} with fronthaul delivery  for all possible $i$, we obtain the NDT for group $(m,1)$ as
\begin{equation}
  \tau_{m,1}=\min_{i\in[2]\cup\{0\}} \left\{\frac{3\binom{3}{m+1}if_{m,1}}{2r}+\frac{3\binom{2}{m}f_{m,1}}{d_{m,1+i}}\right\}.\notag
\end{equation}

Considering all possible $m$ and $n$, the achievable NDT of the $3\times3$ Fog-RAN with decentralized caching is $\sum_{m=0}^2 \sum_{n=0}^3 \tau_{m,n}$, where $\tau_{m,n}$ is the  NDT for group $(m,n)$.

\setcounter{subsubsection}{0}

\section{Achievable Upper Bound of NDT}\label{section upper bound}
Generalizing the achievable scheme in Section \ref{section delivery33} to arbitrary $N_T,N_R \ge 2$, we obtain an achievable upper bound of the minimum NDT in the Fog-RAN with decentralized cache placement in the following theorem.

\begin{theorem}[Achievable NDT] \label{thm 1}
For the cache-aided Fog-RAN with $N_T\ge2$ ENs, each with a cache of normalized size $\mu_T$,  $N_R\ge2$ UEs, each with a cache of normalized size $\mu_R$,  $N\ge N_R$ files, and a wireless fronthaul link with power scaling $r>0$, the minimum NDT achieved by decentralized caching is upper bounded by $\tau_{upper}=\sum_{m=0}^{N_R-1}\sum_{n=0}^{N_T}\tau_{m,n}$, where
\begin{equation}
  \tau_{m,0}=\frac{\binom{N_R}{m+1}f_{m,0}}{r}+\frac{\binom{N_R-1}{m}f_{m,0}}{d_{m,N_T}},\label{eqn tau m0}
\end{equation}
and
\begin{align}
\tau_{m,n}= \min_{i \in [N_T-n]\cup\{0\}} \tau_{m,n}^i, \label{eqn tau mn gamma}
\end{align}
when $n\ge1$, with
\begin{equation}
  \tau_{m,n}^i\!\!=\!\!\frac{\binom{N_R}{m+1}\!\binom{N_T}{n}\!\min\!\big\{\!1,\!\frac{i}{n+1}\!\big\}\!f_{m,n}}{r}\!\!+\!\!\frac{\binom{N_R\!-\!1}{m}\!\binom{N_T}{n}\!f_{m,n}}{d_{m,n+i}}.\label{eqn thm 1}
\end{equation}
Here $f_{m,n}$ is defined in \eqref{eqn subfile size}, and $d_{m,j}$ is the achievable per-user DoF of the $\binom{N_T}{j}\times\binom{N_R}{m+1}$ cooperative X-multicast channel given in \cite[Lemma 1]{mine}.
\end{theorem}

\begin{proof}
Similar to Section \ref{section delivery33}, we assume that UE $q$, for $q\in[N_R]$, desires $W_q$ in the delivery phase. Excluding the locally cached subfiles,  each UE $q$, for $q\in[N_R]$, wants subfiles $\{W_{q,\Phi,\Psi}:\Phi\not\ni q,\Phi\subseteq[N_R],\Psi\subseteq[N_T]\}$. We divide the subfiles wanted by all UEs into different groups according to the size of $\Phi$ and $\Psi$, indexed by $\{(m,n):0\le m\le N_R-1, 0\le n\le N_T,m,n\in\mathbb{Z}\}$, such that subfiles in group $(m,n)$ are cached at $m$ UEs and $n$ ENs. There are $N_R\binom{N_R-1}{m}\binom{N_T}{n}$ subfiles in group $(m,n)$, each with fractional size $f_{m,n}$. Each group of subfiles is delivered individually in the time division manner. Without loss of generality, we present the delivery strategy for an arbitrary group $(m,n)$.  The delivery strategy is also given in Algorithm \ref{algorithm 1}.

\begin{algorithm}[!t]
\caption{Delivery scheme for $N_T\times N_R$ Fog-RAN with wireless fronthaul}\label{algorithm 1}
\begin{algorithmic}[1]
\FOR{$m=0,1,\ldots,N_R-1$}
\FOR{$n=0,1,\ldots,N_T$}
\IF{$n=0$}\STATE Generate coded messages $\{W_{\Phi^+,\emptyset}^{\oplus}\triangleq\bigoplus\limits_{q\in\Phi^+} W_{q,\Phi^+\backslash\{q\},\emptyset}:\Phi^+\subseteq[N_R],|\Phi^+|=m+1\}$, each desired by $m+1$ UEs
\STATE The MBS sends messages  $\{W_{\Phi^+,\emptyset}^\oplus\}$ to all the $N_T$ ENs one by one
\STATE The network topology in the access link is changed into the $\binom{N_T}{N_T}\times\binom{N_R}{m+1}$ cooperative X-multicast channel whose achievable per-user DoF is $d_{m,N_T}$ in \cite[Lemma 1]{mine}\ELSE
\STATE Generate coded messages $\{W_{\Phi^+,\Psi}^{\oplus}\triangleq\bigoplus\limits_{q\in\Phi^+} W_{q,\Phi^+\backslash\{q\},\Psi}:\Phi^+\subseteq[N_R],|\Phi^+|=m+1,\Psi\subseteq[N_T],|\Psi|=n\}$
\STATE Let $i=\arg\min_{i}  \tau_{m,n}^i$ in \eqref{eqn tau mn gamma}
\STATE Split each coded message $W_{\Phi^+,\Psi}^{\oplus}$ into $\binom{N_T-n}{i}$ sub-messages $\{W_{\Phi^+,\Psi}^{\oplus,\Psi^+}\}$, each with fractional size $\frac{f_{m,n}}{\binom{N_T-n}{i}}$ and corresponding to a unique EN set $\Psi^+:|\Psi^+|=n+i, \Psi^+\supseteq\Psi$
\FOR{$\Psi^+\subseteq[N_T],|\Psi^+|=n+i$}
\FOR{$\Phi^+\subseteq[N_R],|\Phi^+|=m+1$}
\IF{$\binom{n+i}{n}\le\binom{n+i}{n+1}$}\STATE The MBS sends sub-messages  $\{W_{\Phi^+,\Psi}^{\oplus,\Psi^+}:\Psi\subseteq\Psi^+,|\Psi|=n\}$ to EN set $\Psi^+$ one by one\ELSE
\STATE The MBS sends coded sub-messages  $\{\bigoplus_{\Psi\subset\Psi'} W_{\Phi^+,\Psi}^{\oplus,\Psi^+}:\Psi'\subseteq\Psi^+,|\Psi'|=n+1,|\Psi|=n\}$ to EN set $\Psi^+$\ENDIF
\STATE ENs in $\Psi^+$ can access $\{W_{\Phi^+,\Psi}^{\oplus,\Psi^+}:\Psi\subseteq\Psi^+,|\Psi|=n\}$ desired by UE set $\Phi^+$.
\ENDFOR
\ENDFOR
\STATE The network topology in the access link is changed into the $\binom{N_T}{n+i}\times\binom{N_R}{m+1}$ cooperative X-multicast channel whose achievable per-user DoF is $d_{m,n+i}$ in \cite[Lemma 1]{mine}
\ENDIF
\ENDFOR
\ENDFOR
\end{algorithmic}
\end{algorithm}

\subsubsection{$n=0$}
Note that each subfile in group $(m,0)$ is desired by one UE, and already cached at $m$ different UEs but none of ENs. Coded multicasting approach can be used, similar to Section \ref{section delivery33 m0}. In specific, the coded messages are given by
\begin{align}
\left\{W_{\Phi^+,\emptyset}^{\oplus}\triangleq\bigoplus_{q\in\Phi^+} W_{q,\Phi^+\backslash\{q\},\emptyset}:\Phi^+\subseteq[N_R],|\Phi^+|=m+1\right\}.\label{eqn m0 message}
\end{align}
Each coded message $W_{\Phi^+,\emptyset}^{\oplus}$ is desired by UE set $\Phi^+$. (If $m=0$, each coded message $W_{\Phi^+,\emptyset}^\oplus$ degenerates to subfile $W_{q,\emptyset,\emptyset}$ for $\Phi^+=\{q\}$.) These messages need to be generated at the MBS and then delivered to UEs via the fronthaul link and the access link. In the fronthaul link, we let the MBS multicast each coded message in \eqref{eqn m0 message} to all the $N_T$ ENs one
by one. The fronthaul NDT is given by
\begin{align}
\tau_F=\frac{\binom{N_R}{m+1}f_{m,0}}{r}.\label{eqn m0 fronthaul}
\end{align}

By such naive multicast transmission in the fronthaul link, each EN now has access to all the coded messages in \eqref{eqn m0 message}, and can cooperatively transmit together in the access link. The access channel thus becomes the $\binom{N_T}{N_T}\times\binom{N_R}{m+1}$ cooperative X-multicast channel with achievable per-user DoF $d_{m,N_T}$ in \cite[Lemma 1]{mine}. Since each UE desires $\binom{N_R-1}{m}$ messages, the access NDT is given by
\begin{align}
 \tau_A=\frac{\binom{N_R-1}{m}f_{m,0}}{d_{m,N_T}}.\label{eqn m0 access}
\end{align}
Combining \eqref{eqn m0 fronthaul} and \eqref{eqn m0 access}, the achievable NDT for the delivery of group $(m,0)$ is
\begin{align}
  \tau_{m,0}=\frac{\binom{N_R}{m+1}f_{m,0}}{r}+\frac{\binom{N_R-1}{m}f_{m,0}}{d_{m,N_T}}.\label{eqn m0}
\end{align}

\subsubsection{$n>0$}
Note that each subfile in group $(m,n)$ is desired by one UE, and already cached at $m$ different UEs and $n$ different ENs. Coded multicasting approach can be used, similar to Section \ref{section delivery33 m1}. In specific, given an arbitrary UE set $\Phi^+$ with size $|\Phi^+|=m+1$ and an arbitrary EN set $\Psi$ with size $n$, each EN in $\Psi$ generates the coded message $W_{\Phi^+,\Psi}^{\oplus}\triangleq\bigoplus_{q\in\Phi^+} W_{q,\Phi^+\backslash\{q\},\Psi}$ desired by all UEs in $\Phi^+$. (If $m=0$, coded message $W_{\Phi^+,\Psi}^\oplus$ degenerates to subfile $W_{q,\emptyset,\Psi}$ for $\Phi^+=\{q\}$.) Through this coded multicasting approach, $m+1$ different subfiles are combined into a single coded message via XOR, and there are only $\binom{N_R}{m+1}\binom{N_T}{n}$ coded messages to be transmitted in total, each available at $n$ ENs and desired by $m+1$ UEs.

Similar to Section \ref{section delivery33 m1}, with the aid of fronthaul, we can allow ENs to access the coded messages of others via the transmission of the MBS in the fronthaul link, thereby enabling chances for more transmission cooperation in the access link. Assume that after the aid of fronthaul transmission, every set of $n + i$ ENs can form a cooperation group in the access link, where $i \in [N_T-n]\cup\{0\}$ is a design parameter.\footnote{If $i=0$, every set of $n$ ENs already forms a cooperation group in the access link, and the coded messages can be delivered to UEs directly in the access link without the use of fornthaul link. The access channel becomes the $\binom{N_T}{n}\times\binom{N_R}{m+1}$ cooperative X-multicast channel in \cite{mine}.} We split each message $W_{\Phi^+,\Psi}^{\oplus}$ into $\binom{N_T-n}{i}$ sub-messages, each with fractional size $f_{m,n}/\binom{N_T-n}{i}$ and corresponding to a distinct EN set $\Psi^+$ with size $n+i$ such that $\Psi^+\supseteq\Psi$. Denote $W_{\Phi^+,\Psi}^{\oplus,\Psi^+}$ as the sub-message in $W_{\Phi^+,\Psi}^{\oplus}$, which is desired by UE set $\Phi^+$, cached at EN set $\Psi$, and corresponding to EN set $\Psi^+$. Each sub-message $W_{\Phi^+,\Psi}^{\oplus,\Psi^+}$ is sent by EN set $\Psi^+$ exclusively in the access link. Then, for an arbitrary EN set $\Psi^+$ with size $n+i$, each EN in $\Psi^+$ needs to access all the sub-messages
\begin{align}
  \left\{W_{\Phi^+,\Psi}^{\oplus,\Psi^+}:\Phi^+\subseteq[N_R],|\Phi^+|=m+1,\Psi\subseteq\Psi^+,|\Psi|=n\right\}.\label{eqn mn 1}
\end{align}

To do this, the MBS choose one of the two methods below to send sub-messages to ENs in the fronthaul link.
\begin{enumerate}
  \item  Fronthaul Transmission without Coded Multicasting: For each EN set $\Psi^+$, the MBS directly sends sub-messages in \eqref{eqn mn 1} one-by-one, and each EN in $\Psi^+$ decodes all the non-cached sub-messages. By this method, the NDT in the fronthaul link is given by
\begin{align}
 \tau_F^1=\frac{1}{r}\binom{N_R}{m+1}\binom{N_T}{n+i}\binom{n+i}{n}\frac{f_{m,n}}{\binom{N_T-n}{i}}.\label{eqn fronthaul 1}
\end{align}
  \item Fronthaul Transmission with Coded Multicasting: Note that each sub-message is already cached at $n$ ENs. The MBS can exploit coded multicasting opportunities in the fronthaul link. In specific, for each EN set $\Psi^+$, the MBS sends coded sub-messages
\begin{align}
  \left\{\bigoplus_{\Psi\subset\Psi'} W_{\Phi^+,\Psi}^{\oplus,\Psi^+}:\Phi^+\subseteq[N_R],|\Phi^+|=m+1,\Psi'\subseteq\Psi^+,\right.\notag\\
  \left.|\Psi'|=n+1,|\Psi|=n\right\}.\notag
\end{align}
For each coded sub-message $\bigoplus_{\Psi\subset\Psi'} W_{\Phi^+,\Psi}^{\oplus,\Psi^+}$, each EN $p$ in $\Psi'$ caches $n$ sub-messages $\{W_{\Phi^+,\Psi}^{\oplus,\Psi^+}:p\in\Psi,\Psi\subset\Psi'\}$, and can decode the non-cached sub-message $\{W_{\Phi^+,\Psi}^{\oplus,\Psi^+}:p\notin\Psi,\Psi\subset\Psi'\}$. By this method the NDT in the fronthaul link is given by
\begin{align}
   \tau_F^2=\frac{1}{r}\binom{N_R}{m+1}\binom{N_T}{n+i}\binom{n+i}{n+1}\frac{f_{m,n}}{\binom{N_T-n}{i}}.\label{eqn fronthaul 2}
\end{align}.
\end{enumerate}

Choosing the smaller one between \eqref{eqn fronthaul 1} and \eqref{eqn fronthaul 2}, the fronthaul NDT is given by
\begin{align}
  \tau_F  =&\frac{1}{r}\binom{N_R}{m+1}\binom{N_T}{n+i}\frac{f_{m,n}}{\binom{N_T-n}{i}}\min\left\{\binom{n+i}{n},\binom{n+i}{n+1}\right\}\notag\\
  =&\binom{N_R}{m+1}\binom{N_T}{n}\min\left\{1,\frac{i}{n+1}\right\}\frac{f_{m,n}}{r}.\label{eqn fronthaul ndt}
\end{align}

Then in the access link, for an arbitrary EN set $\Psi^+$ with size $n+i$, each EN in $\Psi^+$ cooperatively sends sub-messages in \eqref{eqn mn 1}. The access channel is changed to the $\binom{N_T}{n+i}\times\binom{N_R}{m+1}$ cooperative X-multicast channel with achievable per-user DoF $d_{m,n+i}$ in \cite[Lemma 1]{mine}. Since each UE $q$, for $q\in[N_R]$, wants $\binom{N_R-1}{m}\binom{N_T}{n+i}\binom{n+i}{n}$ sub-messages, the access NDT is
\begin{align}
  \tau_A&=\binom{N_R-1}{m}\binom{N_T}{n+i}\frac{\binom{n+i}{n}}{\binom{N_T-n}{i}}\frac{f_{m,n}}{d_{m,n+i}}\notag\\
  &=\binom{N_R-1}{m}\binom{N_T}{n}\frac{f_{m,n}}{d_{m,n+i}}.\label{eqn access ndt}
\end{align}.

Combining \eqref{eqn fronthaul ndt} and \eqref{eqn access ndt} and taking the minimum of NDT over $i$, we obtain the NDT for the delivery of group $(m,n)$ as
\begin{align}
\tau_{m,n}=\min_{i \in [N_T-n]\cup\{0\}} \tau_{m,n}^i,\label{eqn mn}
\end{align}
where
\begin{align}
  \tau_{m,n}^i=&\binom{N_R}{m+1}\binom{N_T}{n}\min\left\{1,\frac{i}{n+1}\right\}\frac{f_{m,n}}{r}\notag\\
  &+\binom{N_R-1}{m}\binom{N_T}{n}\frac{f_{m,n}}{d_{m,n+i}}.\notag
\end{align}

Summing up NDTs in \eqref{eqn m0} and \eqref{eqn mn} for all groups, the total achievable NDT is
\begin{align}
  \tau=\sum_{m=0}^{N_R-1}\sum_{n=0}^{N_T}\tau_{m,n},\notag
\end{align}
which is the same as in Theorem \ref{thm 1}. Thus, Theorem \ref{thm 1} is proved.
\end{proof}
The first and second terms on the right hand side of both \eqref{eqn tau m0} and \eqref{eqn thm 1}  are the fronthaul NDT and the access NDT, respectively. It is clear that the fronthaul NDT decreases as the power scaling $r$ increases. When $r\! \rightarrow\! \infty$, the fronthaul NDT approaches zero, and the overall achievable NDT is dominated by the access NDT, given by
\begin{equation}
\begin{split}
  \lim_{r \rightarrow \infty}\tau_{upper}=& \sum_{m=0}^{N_R-1}\sum_{n=0}^{N_T}\frac{\binom{N_R-1}{m}\binom{N_T}{n}f_{m,n}}{d_{m,N_T}}\notag\\
  =&\sum_{m=0}^{N_R-1}\frac{\binom{N_R-1}{m}\mu_R^m(1-\mu_R)^{N_R-m}}{d_{m,N_T}},\notag
\end{split}
\end{equation}
which is equivalent to the NDT when $\mu_T\!=\!1$. This means that when the fronthaul capacity is large enough, the fronthaul transmission time can be ignored and hence each EN can access the entire database as when $\mu_T\! =\! 1$. The detailed discussion of the achievable NDT in Theorem \ref{thm 1} and its comparison to \cite{girgis2017decentralized,simeone,Cran,Joan2018storage} are given in Section \ref{section discussion}.

\setcounter{subsubsection}{0}

\section{Lower Bound of NDT}\label{section converse}
In this section, we present a lower bound of the minimum NDT, based on which we show that the achievable scheme is order-optimal.

\begin{theorem}[Lower bound of NDT]\label{thm 2}
For the cache-aided Fog-RAN with $N_T\ge2$ ENs, each with a cache of normalized size $\mu_T$,  $N_R\ge2$ UEs, each with a cache of normalized size $\mu_R$,  $N\ge N_R$ files, and a wireless fronthaul link with power scaling $r>0$, the minimum NDT achieved by decentralized caching is lower bounded by
\begin{align}
  \tau_{lower}\!=\!\max_{l_1\in[N_R]}\!\frac{l_1\!(1\!-\!\mu_T\!)^{N_T}(1\!-\!\mu_R\!)^{l_1}}{r}\!+\!\max_{l_2\in[N_R]}\!\frac{l_2(1\!-\!\mu_R\!)^{l_2}}{\min\{l_2,\!N_T\!\}}.\label{eqn thm 2}
\end{align}
\end{theorem}

\begin{proof}
Since this is the proof of a lower bound, we focus on a specific UE demand that each UE $q$ ($q\in[N_R]$) wants file $W_q$. Since ENs are assumed to be half-duplex, we will prove the lower bound of fronthaul NDT and access NDT separately.

\subsubsection{Fronthaul Transmission}
We first consider the fronthaul transmission. Consider the transmission of the files desired by the first $l_1$ UEs, for $l_1\in[N_R]$. The proof is based on the following observation. Given received signals $Q_{1\sim N_T}$ from the MBS at all ENs and the caches $U_{1\sim N_T}$ at all ENs, one can construct the transmitted signals of all ENs. Then, given all the transmitted signals from the ENs and caches $V_{1\sim l_1}$ at the first $l_1$ UEs, one can obtain the desired files of these UEs almost surely. We have
\begin{align}
H(W_{1\sim l_1}|Q_{1\sim N_T},U_{1\sim N_T},V_{1\sim l_1})=F\varepsilon_F+T_F\varepsilon_P\log P,\notag
\end{align}
where $W_{1\sim l_1}$ are files $\{W_1,W_2,\ldots,W_{l_1}\}$. Here, $\varepsilon_F$ and $\varepsilon_P$ are a function of file size $F$ and a function of power $P$, respectively, and satisfy $\lim_{F\to\infty}\varepsilon_F=0$, $\lim_{P\to\infty}\varepsilon_P=0$. Then, we have
\begin{subequations}\label{eqn converse 1}
\begin{align}
  l_1F=&H(W_{1\sim l_1}|W_{l_1+1\sim N})\label{eqn converse 11}\\
  =&I(W_{1\sim l_1};Q_{1\sim N_T},U_{1\sim N_T},V_{1\sim l_1}|W_{l_1+1\sim N})\notag\\
  &+H(W_{1\sim l_1}|Q_{1\sim N_T},U_{1\sim N_T},V_{1\sim l_1},W_{l_1+1\sim N})\label{eqn converse 12}\\
  =&H(Q_{1\sim N_T},U_{1\sim N_T},V_{1\sim l_1}|W_{l_1+1\sim N})\notag\\
  &-H(Q_{1\sim N_T},U_{1\sim N_T},V_{1\sim l_1}|W_{1\sim N})\notag\\
  &+F\varepsilon_F+T_F\varepsilon_P\log P\label{eqn converse 13}\\
  \le&H(Q_{1\sim N_T},U_{1\sim N_T},V_{1\sim l_1}|W_{l_1+1\sim N})\notag\\
  &+F\varepsilon_F+T_F\varepsilon_P\log P\label{eqn converse 14}\\
  \le& h(Q_{1\sim N_T})+H(U_{1\sim N_T},V_{1\sim l_1}|W_{l_1+1\sim N})\notag\\
  &+F\varepsilon_F+T_F\varepsilon_P\log P,\label{eqn converse 15}
\end{align}
\end{subequations}
where $W_{l_1+1\sim N}$ are files $\{W_{l_1+1},W_{l_1+2},\ldots,W_N\}$. Here, \eqref{eqn converse 12} and \eqref{eqn converse 13} come from the definition of mutual information; \eqref{eqn converse 15} comes from the fact that conditioning reduces entropy. In \eqref{eqn converse 15}, $h(Q_{1\sim N_T})$ is bounded by
\begin{subequations}\label{eqn converse 2}
\begin{align}
  h(Q_{1\sim N_T})=&I(Q_{1\sim N_T};S)+h(Q_{1\sim N_T}|S)\label{eqn converse 21}\\
  =&I(Q_{1\sim N_T};S)+T_F\varepsilon_P\log P\label{eqn converse 22}\\
  \le&T_F(r\log P+\varepsilon_P\log P)+T_F\varepsilon_P\log P.\label{eqn converse 23}
\end{align}
\end{subequations}
Here, $S$ is the transmitted signal of the MBS; \eqref{eqn converse 22} is due to the fact that the conditional entropy $h(Q_{1\sim N_T}|S)$ comes from the noise received at ENs; \eqref{eqn converse 23} follows from the capacity bound of the broadcast channel in high SNR regime. In \eqref{eqn converse 15}, $H(U_{1\sim N_T},V_{1\sim l_1}|W_{l_1+1\sim N})$ is given by
\begin{subequations}\label{eqn converse 3}
\begin{align}
  &H(U_{1\sim N_T},V_{1\sim l_1}|W_{l_1+1\sim N})\notag\\
  =&H(U_{1\sim N_T}^{1\sim l_1},V_{1\sim l_1}^{1\sim l_1})\label{eqn converse 31}\\
  =&\sum_{n=1}^{l_1} H(U_{1\sim N_T}^n,V_{1\sim l_1}^n)\label{eqn converse 32}\\
  =&l_1F\cdot [1-(1-\mu_T)^{N_T}(1-\mu_R)^{l_1}].\label{eqn converse 33}
\end{align}
\end{subequations}
Here, $U_{1\sim N_T}^{1\sim l_1}$, $V_{1\sim l_1}^{1\sim l_1}$ are the cached contents of files $\{W_1,W_2,\ldots,W_{l_1}\}$ at all the $N_T$ ENs and UEs $\{1,2,\ldots, l_1\}$, respectively, and $U_{1\sim N_T}^n,V_{1\sim l_1}^n$ are the cached contents of file $n$ at all the $N_T$ ENs and UEs $\{1,2,\ldots, l_1\}$, respectively; \eqref{eqn converse 31} and \eqref{eqn converse 32} come from the fact that only the cached contents of files $\{W_1,\ldots,W_{l_1}\}$ are unknown given files $\{W_{l_1+1},\ldots,W_N\}$ and that the cache scheme does not allow intra-file coding or inter-file coding; \eqref{eqn converse 33} comes from the fact that each EN and each UE caches a subset of $\mu_TF$ and $\mu_RF$ bits of each file independently and uniformly at random, respectively.

Combining \eqref{eqn converse 15}\eqref{eqn converse 23}\eqref{eqn converse 33}, and letting $F\rightarrow\infty$, $P\rightarrow\infty$, we obtain that
\begin{align}
  \lim_{P\rightarrow\infty}\lim_{F\rightarrow\infty}\frac{T_F\log P}{F}\ge \frac{1}{r}l_1(1-\mu_T)^{N_T}(1-\mu_R)^{l_1}.\label{eqn converse 4}
\end{align}

\subsubsection{Access Phase}
Next we consider the access transmission. The proof method is an extension of the approach in \cite[Section VI]{niesen} by taking decentralized cache scheme into account. Consider the first $l_2$ UEs, for $l_2\in[N_R]$. The proof is based on the following observation. Given the received signals $Y_{1\sim l_2}$ and the cached contents $V_{1\sim l_2}$ of the $l_2$ UEs, one can successfully decode the desired files of these $l_2$ UEs. Thus, we have
\begin{align}
  H(W_{1\sim l_2}|Y_{1\sim l_2},V_{1\sim l_2})=F\varepsilon_F.\notag
\end{align}
Similar to \eqref{eqn converse 1}, we have
\begin{subequations}\label{eqn converse 5}
\begin{align}
  l_2F=&H(W_{1\sim l_2}|W_{l_2+1\sim N})\label{eqn converse 51}\\
  =&I(W_{1\sim l_2};Y_{1\sim l_2},V_{1\sim l_2}|W_{l_2+1\sim N})\notag\\
  &+H(W_{1\sim l_2}|Y_{1\sim l_2},V_{1\sim l_2},W_{l_2+1\sim N})\label{eqn converse 52}\\
  =& I(W_{1\sim l_2};Y_{1\sim l_2},V_{1\sim l_2}|W_{l_2+1\sim N})+F\varepsilon_F\label{eqn converse 53}\\
  =&H(Y_{1\sim l_2},V_{1\sim l_2}|W_{l_2+1\sim N})\notag\\
  &-H(Y_{1\sim l_2},V_{1\sim l_2}|W_{1\sim N})+F\varepsilon_F\label{eqn converse 54}\\
  \le& H(Y_{1\sim l_2},V_{1\sim l_2}|W_{l_2+1\sim N})+F\varepsilon_F\label{eqn converse 55}\\
  \le& h(Y_{1\sim l_2})+H(V_{1\sim l_2}|W_{l_2+1\sim N})+F\varepsilon_F.\label{eqn converse 56}
\end{align}
\end{subequations}
In \eqref{eqn converse 56}, $h(Y_{1\sim l_2})$ is bounded by
\begin{subequations}\label{eqn converse 6}
\begin{align}
  h(Y_{1\sim l_2})=&I(Y_{1\sim l_2};X_{1\sim N_T})+h(Y_{1\sim l_2}|X_{1\sim N_T})\label{eqn converse 61}\\
  =&I(Y_{1\sim l_2};X_{1\sim N_T})+T_F\varepsilon_P\log P\label{eqn converse 62}\\
  \le&T_A\min\{N_T,l_2\}(\log P+\varepsilon_P\log P)\notag\\
  &+T_A\varepsilon_P\log P.\label{eqn converse 63}
\end{align}
\end{subequations}
Here, $X_{1\sim N_T}$ are the transmitted signals from all the $N_T$ ENs; \eqref{eqn converse 62} is due to the fact that the conditional entropy $h(Y_{1\sim l_2}|X_{1\sim N_T})$ comes from the noise received at UEs; \eqref{eqn converse 63} follows from the capacity bound of the $N_T\times l_2$ MIMO channel in high SNR regime, similar to the proof of \cite[Lemma 5]{niesen}.

In \eqref{eqn converse 56}, $H(V_{1\sim l_2}|W_{l_2+1\sim N})$ is given by
\begin{subequations}\label{eqn converse 7}
\begin{align}
  H(V_{1\sim l_2}|W_{l_2+1\sim N})=&H(V_{1\sim l_2}^{1\sim l_2})\label{eqn converse 71}\\
  =&\sum_{n=1}^{l_2} H(V_{1\sim l_2}^n)\label{eqn converse 72}\\
  =&l_2F\cdot [1-(1-\mu_R)^{l_2}].\label{eqn converse 73}
\end{align}
\end{subequations}
Note that \eqref{eqn converse 7} is similar to \eqref{eqn converse 3}, and the detailed explanation is omitted here.

Combining \eqref{eqn converse 56}\eqref{eqn converse 63}\eqref{eqn converse 73}, and letting $F\rightarrow\infty$, $P\rightarrow\infty$, we obtain that
\begin{align}
  \lim_{P\rightarrow\infty}\lim_{F\rightarrow\infty}\frac{T_A\log P}{F}\ge \frac{l_2(1-\mu_R)^{l_2}}{\min\{l_2,N_T\}}.\label{eqn converse 8}
\end{align}

Combining \eqref{eqn converse 4} and \eqref{eqn converse 8}, and taking the maximum over $l_1,l_2\in[N_R]$, the minimum NDT $\tau$ is lower bounded by
\begin{align}
  \tau&=\lim_{P\rightarrow\infty}\lim_{F\rightarrow\infty}\frac{(T_F+T_A)\log P}{F}\notag\\
  &\ge\max_{l_1\in[N_R]}\frac{l_1}{r}(1-\mu_T)^{N_T}(1-\mu_R)^{l_1}\!+\!\max_{l_2\in[N_R]}\frac{l_2(1-\mu_R)^{l_2}}{\min\{l_2,N_T\}},\notag
\end{align}
which finishes the proof of Theorem \ref{thm 2}.
\end{proof}

Comparing Theorem \ref{thm 1} and Theorem \ref{thm 2}, the multiplicative gap between the upper and lower bounds is given in the following corollary, whose proof is in appendix.

\begin{corollary}[Gap of NDT]\label{coro gap}
The multiplicative gap between the upper and lower bounds of the minimum NDT of the considered system is within 12.
\end{corollary}

\section{Numerical Results and Discussions}\label{section discussion}

\begin{figure}[tbp]
\begin{centering}
\includegraphics[scale=0.34]{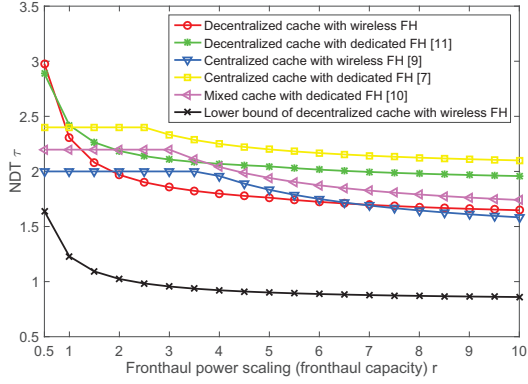}
\caption{NDT when $N_T=2,N_R=5,\mu_T=1/2,\mu_R=1/5$.}\label{Fig compare1}
\end{centering}
\end{figure}

%\begin{figure}[tbp]
%\begin{centering}
%\includegraphics[scale=0.29]{compare4.eps}
%\vspace{-0.2cm}
%\caption{NDT when $N_T=3,N_R=3,\mu_T=1/3,\mu_R=1/3$.}\label{Fig compare2}
%\vspace{-0.8cm}
%\end{centering}
%\end{figure}

In this section, we use numerical examples to compare the achievable NDT of our proposed scheme using decentralized caching with existing schemes, including  centralized caching in \cite{simeone}\footnote{Since \cite{simeone} only considers EN caches, we add conventional uncoded caching at UEs in the plot of \cite{simeone} for fair comparison.},\cite{Cran}, mixed caching (centralized at ENs and decentralized at UEs) in \cite{Joan2018storage}, and decentralized caching (for two ENs only) in \cite{girgis2017decentralized}. Note that \cite{Cran} assumed wireless fronthaul, while \cite{girgis2017decentralized,simeone,Joan2018storage} assumed dedicated fronthaul with capacity of $C_F=r\log P$ bits per symbol. Fig.~\ref{Fig compare1} depicts the NDT when $N_T\!=\!2,N_R\!=\!5,\mu_T\!=\frac{1}{2},\mu_R\!=\frac{1}{5}$. It is seen that when $r$ increases, the achievable NDT, as well as our proposed lower bound, decreases and finally approaches a constant as expected.

%\footnote{{\color{blue}Though this paper only considers decentralized caching, one can still apply centralized caching based on our proposed delivery strategies and obtain an achievable NDT as an optimal solution of a linear programming problem with respect to the subfile sizes as in \cite{mine}. Though this may achieve a smaller NDT, it imposes practical constraints on ENs and UEs, and thus is beyond the scope of this paper.}}

Comparing to the decentralized caching in \cite{girgis2017decentralized}, it is seen in Fig. \ref{Fig compare1} that our achievable NDT is better in most fronthaul capacity regions even though dedicated fronthaul link is considered in \cite{girgis2017decentralized}. This is because EN cooperation in the access link is fully exploited by the careful design of fronthaul transmission in our scheme, while it is only exploited in \cite{girgis2017decentralized} when transmitting some specific subfiles.

When $r$ is small, it is seen that our scheme with decentralized caching is inferior to the centralized caching in \cite{Cran} as expected, since \cite{Cran} can better utilize the cache memory. However, when $r\!\ge\! 2$, our scheme performs very close to \cite{Cran}, and even outperforms it when $2\le r\le6$. This is because 1) an additional layer of coded multicasting opportunities based on the coded messages desired by UEs is exploited in the fronthaul link in our scheme as stated in Section \ref{section delivery33 m1}, while in \cite{Cran}, the coded multicasting opportunities in the fronthaul link are only exploited by generating coded messages directly from requested subfiles $\{W_{q,\Phi,\Psi}\}$; 2) we obtain a larger achievable per-user DoF than the one in \cite{Cran} in the access channel by using interference neutralization and interference alignment jointly.

Comparing to the centralized caching in \cite{simeone} and mixed caching in \cite{Joan2018storage} with dedicated fronthaul, it is seen that our achievable NDT is even better than theirs when $r\ge2$. This is because in their schemes, the fronthaul link is not used to deliver contents already cached at ENs to boost EN cooperation in the access link. Note that the authors in \cite{Koh2017Cloud} state that coded multicasting is not useful in certain cases in the fronthaul delivery. However, by comparing to \cite{simeone} with conventional uncoded caching at UEs, it is still seen that coded multicasting plays an important role in the access delivery in our scheme to reduce the sum NDT.

\section{Conclusion}
In this paper, we investigated the latency performance of a Fog-RAN with wireless fronthaul and for arbitrary number of cache-equipped ENs and cache-equipped UEs, by using NDT as the performance metric. The system consists of two phases: a file-splitting based decentralized cache placement phase and a fronthaul-aided two-hop content delivery phase. In our proposed delivery scheme, coded multicasting opportunities are fully exploited in the fronthaul link by fetching both non-cached and cached contents of ENs to enhance EN cooperation in the access transmission. Then, the access link is changed into the cooperative X-multicast channels. We obtained both the achievable upper bound and theoretical lower bound of NDT for decentralized caching, with a multiplicative gap less than 12. It is shown that our decentralized caching scheme can balance the time between the fronthaul link and access link by the careful design of fronthaul transmission, and even outperforms the centralized schemes, mixed scheme, and decentralized scheme (with dedicated fronthaul) under certain conditions.

\setcounter{subsection}{0}
\newpage
\section*{Appendix: Proof of Corollary \ref{coro gap}}
In Appendix, we aim to prove Corollary \ref{coro gap}. We consider two cases to prove the gap, i.e., $N_T\ge N_R$ and $N_T<N_R$.

\subsection{$N_T\ge N_R$}
We first consider the case when $N_T\ge N_R$. The achievable upper bound of NDT is $\tau_{upper}=\sum_{m=0}^{N_R-1}\sum_{n=0}^{N_T}\tau_{m,n}$, where $\tau_{m,n}$ is given in \eqref{eqn tau m0} and \eqref{eqn tau mn gamma}. Taking $i=0$ in \eqref{eqn tau mn gamma}, $\tau_{m,n}$ ($n>0$) is bounded by
\begin{align}
  \tau_{m,n}\le\frac{\binom{N_R-1}{m}\binom{N_T}{n}f_{m,n}}{d_{m,n}}.\notag
\end{align}
We also have
\begin{align}
  \tau_{m,0}\!=\!\binom{N_R}{m+1}\frac{f_{m,0}}{r}\!+\!\frac{\binom{N_R-1}{m}f_{m,0}}{d_{m,N_T}}.\notag
\end{align}
When $N_T\ge N_R$, it is easy to see that $d_{m,n}\ge1/2$ for $m\in[N_R-1]\cup\{0\},n\in[N_T]$. Then, $\tau_{upper}$ is upper bounded by \eqref{eqn gap 1}.
\begin{figure*}[t]
\begin{align}
\tau_{upper}\le&\sum_{m=0}^{N_R-1}\sum_{n=1}^{N_T}\frac{\binom{N_R-1}{m}\binom{N_T}{n}f_{m,n}}{d_{m,n}}+\sum_{m=0}^{N_R-1}\frac{\binom{N_R-1}{m}f_{m,0}}{d_{m,N_T}}+\frac{1}{r}\sum_{m=0}^{N_R-1}\binom{N_R}{m+1}f_{m,0}\notag\\
\le&2\sum_{m=0}^{N_R-1}\sum_{n=0}^{N_T}\binom{N_R-1}{m}\binom{N_T}{n}f_{m,n}+\frac{(1-\mu_T)^{N_T}}{r}\sum_{m=0}^{N_R-1}\binom{N_R}{m+1}\mu_R^m(1-\mu_R)^{N_R-m}\notag\\
=&2(1-\mu_R)+\frac{(1-\mu_T)^{N_T}}{r}\frac{1-\mu_R}{\mu_R}\sum_{m=0}^{N_R-1}\binom{N_R}{m+1}\mu_R^{m+1}(1-\mu_R)^{N_R-m-1}\notag\\
=&2(1-\mu_R)+\frac{(1-\mu_T)^{N_T}}{r}\frac{1-\mu_R}{\mu_R}\sum_{p=1}^{N_R}\binom{N_R}{p}\mu_R^{p}(1-\mu_R)^{N_R-p}\notag\\
=&2(1-\mu_R)+\frac{(1-\mu_T)^{N_T}}{r}\frac{1-\mu_R}{\mu_R}\left[\sum_{p=0}^{N_R}\binom{N_R}{p}\mu_R^{p}(1-\mu_R)^{N_R-p}-(1-\mu_R)^{N_R}\right]\notag\\
=&2(1-\mu_R)+\frac{(1-\mu_T)^{N_T}}{r}\frac{1-\mu_R}{\mu_R}\left[1-(1-\mu_R)^{N_R}\right]\label{eqn gap 1}
\end{align}
\hrule
\end{figure*}
Taking $l_2=1$ in \eqref{eqn thm 2}, the lower bound of NDT is lower bounded by
\begin{align}
  \tau_{lower}\ge\max_{l_1}\frac{l_1(1-\mu_T)^{N_T}(1-\mu_R)^{l_1}}{r}+(1-\mu_R).\label{eqn gap 2}
\end{align}
Denote $g$ as the multiplicative gap, then the gap is bounded by
\begin{align}
  g\le\frac{2(1-\mu_R)+\frac{(1-\mu_T)^{N_T}}{r}\frac{1-\mu_R}{\mu_R}\left[1-(1-\mu_R)^{N_R}\right]}{\max_{l_1}\frac{l_1(1-\mu_T)^{N_T}(1-\mu_R)^{l_1}}{r}+(1-\mu_R)}.\notag
\end{align}
To upper bound $g$, we first consider
\begin{align}
  g_F\triangleq&\frac{\frac{(1-\mu_T)^{N_T}}{r}\frac{1-\mu_R}{\mu_R}\left[1-(1-\mu_R)^{N_R}\right]}{\max_{l_1}\frac{l_1(1-\mu_T)^{N_T}(1-\mu_R)^{l_1}}{r}}\notag\\
  =&\frac{\frac{1-\mu_R}{\mu_R}\left[1-(1-\mu_R)^{N_R}\right]}{\max_{l_1}l_1(1-\mu_R)^{l_1}},\notag
\end{align}
which can also be viewed as the gap in the fronthaul link. We consider four cases to upper bound $g_F$, i.e. (1) $N_R\le 12$; (2) $N_R\ge 13, \mu_R\ge\frac{1}{12}$; (3) $N_R\ge 13, \frac{1}{N_R}\le\mu_R<\frac{1}{12}$; (4) $N_R\ge13, \mu_R<\frac{1}{N_R}$. Note that the broadcast channel in the fronthaul link is similar to the one-server shared link in \cite{fundamentallimits,decentralized}, and the proof here is similar to the one in  \cite{fundamentallimits,decentralized}.
\subsubsection{$N_R\le 12$}
In this case, using the inequality $(1-\mu_R)^{N_R}\ge1-N_R\mu_R$, we have
\begin{align}
\frac{1-\mu_R}{\mu_R}\left[1-(1-\mu_R)^{N_R}\right]\le\frac{1-\mu_R}{\mu_R}N_R\mu_R\le12(1-\mu_R).\notag
\end{align}
Letting $l_1=1$, $g_F$ is bounded by
\begin{align}
  g_F\le\frac{12(1-\mu_R)}{1-\mu_R}=12.\notag
\end{align}
\subsubsection{$N_R\ge 13, \mu_R\ge\frac{1}{12}$}
We have
\begin{align}
\frac{1-\mu_R}{\mu_R}\left[1-(1-\mu_R)^{N_R}\right]\le\frac{1-\mu_R}{\mu_R}\le12(1-\mu_R).\notag
\end{align}
Similar to Case 1 that $N_R\le12$, $g_F$ is also upper bounded by 12.
\subsubsection{$N_R\ge 13, \frac{1}{N_R}\le\mu_R<\frac{1}{12}$}
Letting $l_1=\lfloor\frac{1}{4\mu_R}\rfloor$, we have
\begin{align}
  \max_{l_1}l_1(1-\mu_R)^{l_1}&\ge\lfloor\frac{1}{4\mu_R}\rfloor(1-\mu_R)^{\lfloor\frac{1}{4\mu_R}\rfloor}\notag\\
  &\ge \lfloor\frac{1}{4\mu_R}\rfloor(1-\lfloor\frac{1}{4\mu_R}\rfloor\mu_R)\notag\\
  &\ge (\frac{1}{4\mu_R}-1)(1-\frac{1}{4\mu_R}\mu_R)\notag\\
  &=\frac{3}{16\mu_R}-\frac{3}{4}.\notag
\end{align}
Then, $g_F$ is upper bounded by
\begin{align}
  g_F\le&\frac{\frac{1-\mu_R}{\mu_R}\left[1-(1-\mu_R)^{N_R}\right]}{\frac{3}{16\mu_R}-\frac{3}{4}}\notag\\
  \le&\frac{1/\mu_R}{\frac{3}{16\mu_R}-\frac{3}{4}}\notag\\
  =&\frac{1}{3/16-3\mu_R/4}\notag\\
  <&\frac{1}{3/16-3/48}=8.\notag
\end{align}
\subsubsection{$N_R\ge13, \mu_R<\frac{1}{N_R}$}
Letting $l_1=\lfloor\frac{N_R}{4}\rfloor$, we have
\begin{align}
g_F&\le\frac{\frac{1-\mu_R}{\mu_R}\left[1-(1-\mu_R)^{N_R}\right]}{\lfloor\frac{N_R}{4}\rfloor(1-\mu_R)^{\lfloor\frac{N_R}{4}\rfloor}}\notag\\
&=\frac{1-(1-\mu_R)^{N_R}}{\mu_R\lfloor\frac{N_R}{4}\rfloor(1-\mu_R)^{\lfloor\frac{N_R}{4}\rfloor-1}}\notag\\
&\le\frac{1-(1-N_R\mu_R)}{\mu_R\lfloor\frac{N_R}{4}\rfloor(1-\mu_R)^{\lfloor\frac{N_R}{4}\rfloor-1}}\notag\\
&=\frac{N_R}{\lfloor\frac{N_R}{4}\rfloor}\frac{1}{(1-\mu_R)^{\lfloor\frac{N_R}{4}\rfloor-1}}\notag\\
&\le\frac{N_R}{\lfloor\frac{N_R}{4}\rfloor}\frac{1}{1-(\lfloor\frac{N_R}{4}\rfloor-1)\mu_R}\notag\\
&\le\frac{N_R}{\frac{N_R}{4}-1}\frac{1}{1-(\frac{N_R}{4}-1)\mu_R}\notag\\
&<\frac{1}{\frac{1}{4}-\frac{1}{N_R}}\frac{1}{1-(\frac{N_R}{4}-1)\frac{1}{N_R}}\notag\\
&\le\frac{1}{\frac{1}{4}-\frac{1}{13}}\frac{1}{\frac{3}{4}+\frac{1}{N_R}}<8.\notag
\end{align}
Combining all four cases, we find that $g_F\le12$ for all $\mu_R,N_R$. Then, the gap $g$ is upper bounded by
\begin{align}
  g\le\frac{2(1-\mu_R)+12\max_{l_1}\frac{l_1(1-\mu_T)^{N_T}(1-\mu_R)^{l_1}}{r}}{\max_{l_1}\frac{l_1(1-\mu_T)^{N_T}(1-\mu_R)^{l_1}}{r}+(1-\mu_R)}\le12.\notag
\end{align}
Thus, we proved the case when $N_T\ge N_R$.

\subsection{$N_T<N_R$}
Now, we consider the case when $N_T<N_R$. The achievable upper bound of NDT is $\tau_{upper}=\sum_{m=0}^{N_R-1}\sum_{n=0}^{N_T}\tau_{m,n}$, where $\tau_{m,n}$ is given in \eqref{eqn tau m0} and \eqref{eqn tau mn gamma}. Taking $i=0$ in \eqref{eqn tau mn gamma}, $\tau_{m,n}$ ($n>0$) is bounded by
\begin{align}
  \tau_{m,n}\le\frac{\binom{N_R-1}{m}\binom{N_T}{n}f_{m,n}}{d_{m,n}}.\notag
\end{align}
We also have
\begin{align}
  \tau_{m,0}\!=\!\binom{N_R}{m+1}\frac{f_{m,0}}{r}\!+\!\frac{\binom{N_R-1}{m}f_{m,0}}{d_{m,N_T}}.\notag
\end{align}
It is easy to see in \cite[Lemma 1]{mine} that $d_{m,n}\ge d_{m,1}=\frac{N_T}{N_T+\frac{N_R-m-1}{m+1}}$ for $m\in[N_R-1]\cup\{0\},n\in[N_T]$. Then, the achievable upper bound of NDT is bounded by \eqref{eqn gap 3}.
\begin{figure*}[t]
\begin{align}
  \tau_{upper}\le&\sum_{m=0}^{N_R-1}\sum_{n=1}^{N_T}\frac{\binom{N_R-1}{m}\binom{N_T}{n}f_{m,n}}{d_{m,n}}+\sum_{m=0}^{N_R-1}\frac{\binom{N_R-1}{m}f_{m,0}}{d_{m,N_T}}+\frac{1}{r}\sum_{m=0}^{N_R-1}\binom{N_R}{m+1}f_{m,0}\notag\\
  \le&\sum_{m=0}^{N_R-1}\sum_{n=1}^{N_T}\frac{\binom{N_R-1}{m}\binom{N_T}{n}f_{m,n}}{\frac{N_T}{N_T+\frac{N_R-m-1}{m+1}}}+\sum_{m=0}^{N_R-1}\frac{\binom{N_R-1}{m}f_{m,0}}{\frac{N_T}{N_T+\frac{N_R-m-1}{m+1}}}+\frac{1}{r}\sum_{m=0}^{N_R-1}\binom{N_R}{m+1}f_{m,0}\notag\\
  =&\sum_{m=0}^{N_R-1}\frac{\binom{N_R-1}{m}}{\frac{N_T}{N_T+\frac{N_R-m-1}{m+1}}}\sum_{n=0}^{N_T}\binom{N_T}{n}f_{m,n}+\frac{1}{r}\sum_{m=0}^{N_R-1}\binom{N_R}{m+1}f_{m,0}\notag\\
  =&\sum_{m=0}^{N_R-1}\frac{\binom{N_R-1}{m}}{\frac{N_T}{N_T+\frac{N_R-m-1}{m+1}}}\mu_R^m(1-\mu_R)^{N_R-m}+\frac{1}{r}\sum_{m=0}^{N_R-1}\binom{N_R}{m+1}f_{m,0}\notag\\
  =&\frac{N_T-1}{N_T}\sum_{m=0}^{N_R-1}\binom{N_R-1}{m}\mu_R^m(1-\mu_R)^{N_R-m}+\frac{1}{N_T}\sum_{m=0}^{N_R-1}\binom{N_R}{m+1}\mu_R^m(1-\mu_R)^{N_R-m}+\frac{1}{r}\sum_{m=0}^{N_R-1}\binom{N_R}{m+1}f_{m,0}\notag\\
  =&\frac{N_T-1}{N_T}(1-\mu_R)\sum_{m=0}^{N_R-1}\binom{N_R-1}{m}\mu_R^m(1-\mu_R)^{N_R-m-1}+\frac{1-\mu_R}{N_T\mu_R}\sum_{m=0}^{N_R-1}\binom{N_R}{m+1}\mu_R^{m+1}(1-\mu_R)^{N_R-m-1}\notag\\
  &+\frac{1}{r}\sum_{m=0}^{N_R-1}\binom{N_R}{m+1}f_{m,0}\notag\\
  =&\frac{N_T-1}{N_T}(1-\mu_R)+\frac{1-\mu_R}{N_T\mu_R}\sum_{p=1}^{N_R}\binom{N_R}{p}\mu_R^{p}(1-\mu_R)^{N_R-p}+\frac{1}{r}\sum_{m=0}^{N_R-1}\binom{N_R}{m+1}f_{m,0}\notag\\
  =&\frac{N_T-1}{N_T}(1-\mu_R)+\frac{1-\mu_R}{N_T\mu_R}\left[\sum_{p=0}^{N_R}\binom{N_R}{p}\mu_R^{p}(1-\mu_R)^{N_R-p}-(1-\mu_R)^{N_R}\right]+\frac{1}{r}\sum_{m=0}^{N_R-1}\binom{N_R}{m+1}f_{m,0}\notag\\
  =&\frac{N_T-1}{N_T}(1-\mu_R)+\frac{1-\mu_R}{N_T\mu_R}\left[1-(1-\mu_R)^{N_R}\right]+\frac{(1-\mu_T)^{N_T}}{r}\frac{1-\mu_R}{\mu_R}\left[1-(1-\mu_R)^{N_R}\right]\label{eqn gap 3}
\end{align}
\hrule
\end{figure*}
Using Theorem \ref{thm 2}, the multiplicative gap $g$ is bounded by \eqref{eqn gap 4}.
\begin{figure*}[t]
\begin{align}
  g\le\frac{\frac{N_T-1}{N_T}(1-\mu_R)+\frac{1-\mu_R}{N_T\mu_R}\left[1-(1-\mu_R)^{N_R}\right]+\frac{(1-\mu_T)^{N_T}}{r}\frac{1-\mu_R}{\mu_R}\left[1-(1-\mu_R)^{N_R}\right]}{\max_{l_1\in[N_R]}\frac{l_1(1-\mu_T)^{N_T}(1-\mu_R)^{l_1}}{r}+\max_{l_2\in[N_R]}\frac{l_2(1-\mu_R)^{l_2}}{\min\{l_2,N_T\}}}.\label{eqn gap 4}
\end{align}
\hrule
\end{figure*}
In \eqref{eqn gap 4}, from the analysis when $N_T\ge N_R$, we have
\begin{align}
  \frac{\frac{(1-\mu_T)^{N_T}}{r}\frac{1-\mu_R}{\mu_R}\left[1-(1-\mu_R)^{N_R}\right]}{\max_{l_1\in[N_R]}\frac{l_1(1-\mu_T)^{N_T}(1-\mu_R)^{l_1}}{r}}\le12.\notag
\end{align}
Then, to bound $g$ in \eqref{eqn gap 4}, we first consider
\begin{align}
  g_A\triangleq\frac{\frac{N_T-1}{N_T}(1-\mu_R)+\frac{1-\mu_R}{N_T\mu_R}\left[1-(1-\mu_R)^{N_R}\right]}{\max_{l_2\in[N_R]}\frac{l_2(1-\mu_R)^{l_2}}{\min\{l_2,N_T\}}},\notag
\end{align}
which can also be viewed as the multiplicative gap in the access link. We use three cases to upper bound $g_A$, i.e., (1) $\mu_R<\frac{1}{4N_R}$; (2) $\frac{1}{4N_R}\le\mu_R<\frac{1}{4N_T}$; (3) $\mu_R\ge\frac{1}{4N_T}$.
\subsubsection{$\mu_R<\frac{1}{4N_R}$}
Letting $l_2=N_R$, we have
\begin{align}
\max_{l_2\in[N_R]}\frac{l_2(1-\mu_R)^{l_2}}{\min\{l_2,N_T\}}\ge&\frac{N_R(1-\mu_R)^{N_R}}{N_T}\notag\\
  \ge&\frac{N_R(1-N_R\mu_R)}{N_T}\notag\\
  >&\frac{N_R}{N_T}(1-N_R\frac{1}{4N_R})=\frac{3N_R}{4N_T}.\label{eqn gap 5}
\end{align}
Letting $l_2=1$, we have
\begin{align}
  \max_{l_2\in[N_R]}\frac{l_2(1-\mu_R)^{l_2}}{\min\{l_2,N_T\}}\ge1-\mu_R.\label{eqn gap 6}
\end{align}
We also have
\begin{align}
  \frac{1-\mu_R}{N_T\mu_R}\left[1-(1-\mu_R)^{N_R}\right]\le&\frac{1-\mu_R}{N_T\mu_R}\left[1-(1-N_R\mu_R)\right]\notag\\
  =&\frac{N_R(1-\mu_R)}{N_T}\notag\\
  \le&\frac{N_R}{N_T}.\label{eqn gap 7}
\end{align}
Combining \eqref{eqn gap 5}\eqref{eqn gap 6}\eqref{eqn gap 7}, $g_A$ is upper bounded by
\begin{align}
  g_A=&\frac{\frac{N_T-1}{N_T}(1-\mu_R)+\frac{1-\mu_R}{N_T\mu_R}\left[1-(1-\mu_R)^{N_R}\right]}{\max_{l_2} \frac{l_2(1-\mu_R)^{l_2}}{\min\{l_2,N_T\}}}\notag\\
  \le&\frac{\frac{N_T-1}{N_T}(1-\mu_R)}{1-\mu_R}+\frac{\frac{N_R}{N_T}}{\frac{3N_R}{4N_T}}\notag\\
  <&1+4/3=7/3.\notag
\end{align}

\subsubsection{$\frac{1}{4N_R}\le\mu_R<\frac{1}{4N_T}$}
Letting $l_2=\lceil\frac{1}{4\mu_R}\rceil$, we have
\begin{align}
&\max_{l_2} \frac{l_2(1-\mu_R)^{l_2}}{\min\{l_2,N_T\}}\notag\\
  \ge&\frac{\lceil\frac{1}{4\mu_R}\rceil(1-\mu_R)^{\lceil\frac{1}{4\mu_R}\rceil}}{\min\{\lceil\frac{1}{4\mu_R}\rceil,N_T\}}\notag\\
  \ge&\frac{\frac{1}{4\mu_R}(1-\lceil\frac{1}{4\mu_R}\rceil\mu_R)}{N_T}\notag\\
  \ge&\frac{1-(\frac{1}{4\mu_R}+1)\mu_R}{4\mu_RN_T}\notag\\
  =&\frac{\frac{3}{4}-\mu_R}{4N_T\mu_R}\notag\\
  >&\frac{\frac{3}{4}-\frac{1}{8}}{4N_T\mu_R}=\frac{5}{32N_T\mu_R}.\label{eqn gap 8}
\end{align}
We also have
\begin{align}
  \frac{1-\mu_R}{N_T\mu_R}\left[1-(1-\mu_R)^{N_R}\right]\le\frac{1}{N_T\mu_R}.\label{eqn gap 9}
\end{align}
Combining \eqref{eqn gap 6}\eqref{eqn gap 8}\eqref{eqn gap 9}, $g_A$ is upper bounded by
\begin{align}
  g_A=&\frac{\frac{N_T-1}{N_T}(1-\mu_R)+\frac{1-\mu_R}{N_T\mu_R}\left[1-(1-\mu_R)^{N_R}\right]}{\max_{l_2} \frac{l_2(1-\mu_R)^{l_2}}{\min\{l_2,N_T\}}}\notag\\
  \le&\frac{\frac{N_T-1}{N_T}(1-\mu_R)}{1-\mu_R}+\frac{\frac{1}{N_T\mu_R}}{\frac{5}{32N_T\mu_R}}\notag\\
  <&1+32/5=37/5.\notag
\end{align}
\subsubsection{$\mu_R\ge\frac{1}{4N_T}$}
Letting $l_2=\lfloor\frac{1}{4\mu_R}\rfloor$, we have
\begin{align}
  &\max_{l_2} \frac{l_2(1-\mu_R)^{l_2}}{\min\{l_2,N_T\}}\notag\\
  \ge&\frac{\lfloor\frac{1}{4\mu_R}\rfloor(1-\mu_R)^{\lfloor\frac{1}{4\mu_R}\rfloor}}{\min\{\lfloor\frac{1}{4\mu_R}\rfloor,N_T\}}\notag\\
  =&(1-\mu_R)^{\lfloor\frac{1}{4\mu_R}\rfloor}\notag\\
  \ge&1-\lfloor\frac{1}{4\mu_R}\rfloor\mu_R\notag\\
  \ge&1-\frac{1}{4\mu_R}\mu_R=\frac{3}{4}.\label{eqn gap 10}
\end{align}
We also have
\begin{align}
  \frac{1-\mu_R}{N_T\mu_R}\left[1-(1-\mu_R)^{N_R}\right]\le\frac{1}{N_T\mu_R}\le\frac{1}{N_T\frac{1}{4N_T}}=4.\label{eqn gap 11}
\end{align}
Combining \eqref{eqn gap 6}\eqref{eqn gap 10}\eqref{eqn gap 11}, $g_A$ is bounded by
\begin{align}
   g_A=&\frac{\frac{N_T-1}{N_T}(1-\mu_R)+\frac{1-\mu_R}{N_T\mu_R}\left[1-(1-\mu_R)^{N_R}\right]}{\max_{l_2} \frac{l_2(1-\mu_R)^{l_2}}{\min\{l_2,N_T\}}}\notag\\
   \le&\frac{\frac{N_T-1}{N_T}(1-\mu_R)}{1-\mu_R}+\frac{4}{3/4}\notag\\
   <&1+16/3=19/3.\notag
\end{align}

From the above three cases, we find that $g_A<12$. Then the multiplicative gap $g$ is bounded by \eqref{eqn gap 12}, when $N_T<N_R$.
\begin{figure*}[t]
\begin{align}
  g\le&\frac{\frac{N_T-1}{N_T}(1-\mu_R)+\frac{1-\mu_R}{N_T\mu_R}\left[1-(1-\mu_R)^{N_R}\right]+\frac{(1-\mu_T)^{N_T}}{r}\frac{1-\mu_R}{\mu_R}\left[1-(1-\mu_R)^{N_R}\right]}{\max_{l_1\in[N_R]}\frac{l_1(1-\mu_T)^{N_T}(1-\mu_R)^{l_1}}{r}+\max_{l_2\in[N_R]}\frac{l_2(1-\mu_R)^{l_2}}{\min\{l_2,N_T\}}}\notag\\
  <&\frac{12\max_{l_2\in[N_R]}\frac{l_2(1-\mu_R)^{l_2}}{\min\{l_2,N_T\}}+12\max_{l_1\in[N_R]}\frac{l_1(1-\mu_T)^{N_T}(1-\mu_R)^{l_1}}{r}}{\max_{l_1\in[N_R]}\frac{l_1(1-\mu_T)^{N_T}(1-\mu_R)^{l_1}}{r}+\max_{l_2\in[N_R]}\frac{l_2(1-\mu_R)^{l_2}}{\min\{l_2,N_T\}}}=12\label{eqn gap 12}
\end{align}
\hrule
\end{figure*}

Thus we finished the proof of Corollary \ref{coro gap} that the multiplicative gap is within 12.

\bibliographystyle{IEEEtran}
\bibliography{IEEEabrv,journal}
\end{document}